\documentclass{llncs}

\usepackage[T1]{fontenc}
\usepackage[utf8]{inputenc}
\usepackage{moreverb}
\usepackage{tikz}
\usepackage{amsmath}
\usepackage{amsfonts}
\usepackage{amssymb}
\usepackage{hyperref}
\usepackage{authblk}

\ifx\suppress\undefined
\newcommand{\TODO}[1]{
\typeout{WARNING!!! there is still a TODO left}
\marginpar{\textbf{!TODO: }\emph{#1}}
}
\else
\newcommand{\TODO}[1]{}
\fi

\ifx\suppress\undefined
\newenvironment{todo}[1]{\noindent\rule{.3\textwidth}{1pt}\TODO{#1}\\}{\\\rule{.3\textwidth}{1pt}}
\else

\fi

\ifx\suppress\undefined
\newcommand{\NOTE}[1]{
\typeout{WARNING!!! there are still DRAFT NOTES left}
\marginpar{!DRAFT}\emph{\textbf{DRAFT NOTES:} #1}
}
\else
\newcommand{\NOTE}[1]{}
\fi

\newcommand{\cutrk}{\textup{cutrk}}

\title{Minimum Degree up to Local Complementation: \\Bounds,~Parameterized~Complexity, and~Exact~Algorithms}
\institute{LIG, University of Grenoble, France \and CNRS, Inria project team Carte, LORIA, Nancy, France}
\author{David Cattan\'eo\inst{1}, Simon Perdrix \inst{2}}
\date{}

\begin{document}

\maketitle

\begin{abstract}
The local minimum degree of a graph is the minimum degree that can be reached by means of local complementation. For any $n$, there exist graphs of order $n$ which have a local minimum degree at least $0.189n$, or at least $0.110n$ when restricted to bipartite graphs. 
Regarding the upper bound, we show that 
the local minimum degree is at most $\frac 38n+o(n)$ for general graphs and $\frac n4+o(n)$ for bipartite graphs, improving the known $\frac n2$ upper bound. We also prove that the local minimum degree is smaller than half of the vertex cover number (up to a logarithmic term). 

The local minimum degree problem is NP-Complete and hard to approximate. We show that this problem, even when restricted to bipartite graphs, is in W[2] and FPT-equivalent to the \textsc{EvenSet} problem, whose W[1]-hardness is a long standing open question.  Finally, we show that the local minimum degree is computed by a $\mathcal O^*(1.938^n)$-algorithm, and a $\mathcal O^*(1.466^n)$-algorithm for the bipartite graphs. 

 \end{abstract}

\section{Introduction}

\subsubsection{Notations.} Given a graph $G=(V,E)$, $\sim_G$ denotes the neighbourhood relation of $G$ i.e., $\forall u,v\in V$, $u\sim_G v \Leftrightarrow \{u,v\}\in E$. We consider simple ($\forall u\in V, u\not\sim u$), undirected ($u\sim v \Leftrightarrow v\sim u$) graphs. The set $N_G(u) = \{v ~|~ u \sim_G v\}$ is the neighbourhood of $u$ and its size $\delta_G(u) = |N_G(u)|$ is the degree of $u$. $\delta(G) = \min_{u\in V}\delta_G(u)$ is the minimum degree of $G$ and $\tau(G)$ is the vertex cover number i.e., the size of the smallest set $S$ such that if $u\sim v$, then $u\in S$ or $v\in S$. For any $D\subseteq V$, $Odd_G(D)=\Delta_{u\in D} N_G(u) = \{v\in V~|~|N_G(v)\cap D|=1 \bmod 2\}$ is the odd-neighbourhood of $D$, where $\Delta$ denotes the symmetric difference. 

\subsubsection{Local complementation.} Local complementation of a graph with respect to one of its vertices consists in complementing the neighbourhood of this vertex:

\begin{definition} The local complementation of a graph $G$ with respect to one of its vertices $u$  is the graph $G\star u$ such that $v{\thicksim_{G\star u}} w$ iff $(v{\thicksim_G} w)$ xor $(u{\thicksim_G}v \wedge u{\thicksim_G} w)$.  
\end{definition}

The local complementation is an involution ($G\star u \star u =G$). Two graphs are LC-equivalent if there exists a sequence of local complementation transforming one into the other:  
$G\equiv_{LC} H \Leftrightarrow \exists u_0,\ldots u_k, G\star u_0\ldots \star u_k = H$.

Local complementation has been introduced by Kotzig \cite{kotzig}. 
 The study of this quantity is motivated by several applications: Bouchet \cite{Bouchet90-kappa,Bouchet94-circle} and de Fraysseix \cite{deFraysseix81} used local complementation to give a characterization of circle graphs, and Oum \cite{Oum08} links the notion of {vertex minor of a graph} to LC-equivalence. 
 A noticeable property of local complementation proved by Bouchet \cite{Bouchet87} is that LC-equivalence of graphs can be decided in time polynomial  in the order of the graphs.

\subsubsection{Cut rank.} Local complementation is related to the cut-rank function\footnote{It was used by Bouchet \cite{Bouchet87} and others under the name \emph{connectivity function}, and coined the cut-rank by Oum \cite{Oum08}.} \cite{Bouchet87,Oum08}: given a graph $G$ and a bipartition $(A, V{\setminus} A)$ of its vertices, $\cutrk_G(A)$ is the rank of the linear map $L_A:2^A\to 2^{V\setminus A} = X\mapsto Odd_G(X)\cap (V{\setminus} A)$. 
$L_A$ is linear with respect to the symmetric difference: $L_A(X\Delta Y)=L_A(X)\Delta L_A(Y)$. The cut-rank can equivalently be defined as the rank of the cut-matrix, a sub-matrix of the adjacency matrix. Notice that for any $A$, $\cutrk_G(A)= \cutrk_G(V{\setminus} A)$. 

LC-equivalent graphs have the same cutrank ($\cutrk_G(\cdot) = \cutrk_{G\star u}(\cdot)$) \cite{Bouchet85-connectivity},  however the converse which was conjectured   in \cite{Bouchet87}, has been disproved by Fon deer Flaass \cite{fon1996local}: the counterexample involves two isomorphic Petersen graphs which have the same cut-rank but which are not LC-equivalent.   

\subsubsection{LU-equivalence.} More recently, local complementation has emerged as a key operation in the field of quantum information theory. The graph state formalism consists in representing a quantum state using a graph (see \cite{HEB04} for details). This powerful formalism provides a graphical representation of  quantum entanglement: each vertex represent a quantum bit (qubit) and the edges represent intuitively the entanglement between the qubits. 
Since entanglement is a non local property,  the strength of the entanglement can only decrease when \emph{local} operations are applied on the quantum state, and as a consequence the entanglement is invariant by \emph{local reversible} operations. In the field of quantum information theory this intuition is captured by the LU-equivalence of quantum states: two quantum states have the same entanglement if and only if they are LU-equivalent i.e., there is a local unitary operation transforming one state into the other.  
LU-equivalence of quantum states can be naturally lifted to  graphs as follows: two graphs are LU-equivalent if and only if the corresponding quantum states are LU-equivalent.  
Van den Nest \cite{VdN04} proved that LC-equivalent graphs are LU-equivalent. Moreover Hein et al. \cite{HEB04} proved that LU-equivalent graphs have the same cutrank. Thus LU-equivalence is weaker than LC-equivalence but stronger than the cut-rank equivalence. 
Using Fon der Flaass's  counterexample based on the Petersen graph, one can show that there exist pairs of graphs which are not LU-equivalent but which have the same cutrank \cite{HEB04}.  
LC- and LU-equivalences were conjectured to coincide \cite{Schling05}. Indeed, LC- and LU-equivalence   actually coincide for several families of graphs \cite{VdN05,zeng07}, however a counterexample of order 27 has been discovered using computer assisted methods \cite{CJWY07}.

\subsubsection{Local minimum degree.} In this paper we will focus on the minimum degree up to local complementation called local minimum degree:
 \begin{definition}Given a graph $G$, the \emph{local minimum degree} of $G$ is 
$$\delta_{loc}(G) = \min_{H\equiv_{LC}G}\delta(H)$$  
\end{definition}

The local minimum degree  has been used to bound the rate of some quantum codes obtained by graph concatenation \cite{ChuangShor}. This quantity has also been used to characterise the complexity of preparation of graph states \cite{HMP06} which are used as a resource in measurement-based quantum computation \cite{RB01} (a model of quantum computation which is very promising in terms of physical implementation), as well as blind quantum computation \cite{BFK08} for instance. The local minimum degree is also used to bound the optimal threshold that can be achieved by graph-based quantum secret sharing \cite{MS08,GJMP-MEMICS12}.

The local minimum degree is related to the cut-rank function and the smallest set of the form $D\cup Odd_G (D)$:

\begin{property}[\cite{HMP06}]\label{prop:hmp}
Given a graph $G=(V,E)$, $$\delta_{loc}(G) + 1 = \min_{\emptyset\subset D\subseteq V}|D\cup Odd_G(D)|=\min \{|A|: A\subseteq V\wedge \cutrk_G(A)<|A|\}$$ 
\end{property}

 The second equation provides a cut-rank characterisation of the local minimum degree which implies that two graphs which have the same cut-rank   have the same local minimum degree. As a consequence, since LU-equivalent graphs have the same cut-rank function, they have the same local minimum degree, too.  Thus the local minimum degree is invariant for the three closely related, albeit distinct, classes of equivalence based respectively on local complementation, local unitary operations, and cut-rank functions.

\subsubsection{Bounds on the local minimum degree.} The local minimum degree has been studied for several families of graphs: the local minimum degree of the hypercube is at least logarithmic in the order of the hypercube \cite{HMP06}; the local minimum degree of a Paley graph $\mathcal P_n$ of order $n$ is at least $\sqrt{n}$. There is no known specific upper bound on the local minimum degree of Paley graphs except that not all Paley graphs can have a linear local minimum degree (i.e., $\delta_{loc} (\mathcal P_n)= \Theta(n)$), and the existence of an infinite number of Paley graphs with a linear local minimum degree would imply the Bazzi-Mitter  conjecture on elliptic curves \cite{Javelle-these,JMP-WG12}. 

There is no known explicit construction which leads to a local minimum degree greater than the square root of the order of the graph, however using probabilistic methods, it has been proven that there exist graphs of order $n$ which have a local minimum degree larger than $0.189n$ \cite{JMP-WG12}. There are even bipartite graphs with a linear local minimum degree: for any $n$ there exists a bipartite graph of order $n$ and local minimum degree at least $0.110n$  \cite{JMP-WG12}. 

Regarding the upper-bounds, Property \ref{prop:hmp} implies that the local minimum degree is at most half of the order of the graph, since no set  larger than half of the vertices can have a full cut-rank. In section \ref{sec:bounds}, 
we improve this upper bound, proving that for any graph of order $n$, its local minimum degree is at most  $\frac38 n +o(n)$, and $\frac n4 + o(n)$ for bipartite graphs. We also prove that the local minimum degree is smaller than half of the vertex cover number (up to a logarithmic term).

\subsubsection{Complexity of the local minimum degree.} One motivation for studying the complexity of computing the local minimum degree comes from the problem of producing graphs with a `large' local minimum degree. Indeed,  there is no known explicit construction of graphs with a local minimum degree linear in the order of the graph, but a random graph has such a `large' local minimum degree with high probability. So to produce a graph with a large local minimum degree, one can pick a graph at random and then double check that the local minimum degree is actually `large'. However, computing the local minimum degree is hard, even for bipartite graphs: the associated decision problem is NP-Complete \cite{JMP-WG12} and hard to approximate \cite{JMP-WG12}. 

In section \ref{sec:FPT}, we investigate the parameterized complexity of the local minimum degree problem and its restriction to bipartite graphs. We show that both problems are  FPT-equivalent to the so-called \textsc{EvenSet} problem , implying their W[2]-membership. However, it does not imply any hardness result since the W[1]-hardness of EvenSet is long standing open question \cite{oddset}. 

In section \ref{sec:expo}, we  introduce exponential algorithms for computing the local minimum degree, mainly based on the improved upper bounds. We show that the local minimum degree of any graph of order $n$ can be computed in time $\mathcal O^*(1.938^n)$ and more interestingly that the local minimum degree of bipartite graphs can be computed in time $\mathcal O^*(1.466^n)$.

\section{Upperbounds on  the local minimum degree.}\label{sec:bounds}

For improving the known bounds on the local minimum degree, 
we use as a routine the fact that in any bipartite graph $G=(V_1,V_2,E)$, there exists a non empty subset of $V_1$ which oddly dominates at most $\frac{|V_2|}{2(1-2^{-|V_1|})}$ vertices, so roughly speaking as long as $V_1$ is not too small with respect to $V_2$ there is a non empty subset of $V_1$ which oddly dominates at most half of the vertices of $V_2$. This fact is a direct consequence of the so called Plotkin bound \cite{plotkin60} on linear codes:

\begin{lemma}\label{lem:odd}
For any bipartite graph $G=(V_1,V_2,E)$, there exists a non empty set $D \subseteq V_1$ s.t. $$|Odd_G(D)|\le \frac{|V_2|}{2(1-2^{-|V_1|})}$$
\end{lemma}

\begin{proof}
$C :=\{Odd_G(D) : D\subseteq V_1\}$ is a linear binary code of length $n=|V_2|$ and rank $k=|V_1|$, where $Odd_G(D)$ is identified with its indicator vector in $V_2$. According to the Plotkin bound \cite{plotkin60}, the minimum distance $d$ of $C$ is at most $n/(2(1-2^{-k}))$, thus there exists a non empty set $D\subseteq V_1$ such that $|Odd_G(D)|\le |V_2|/(2(1-2^{-|V_1|}))$.\hfill $\Box$%
\end{proof}

The local minimum degree can be bounded by the vertex cover number as follows:

\begin{lemma}\label{lem:tau}
Given a graph $G$ of order $n$ and  vertex cover number $\tau(G)>0$,
$$2\delta_{loc}(G)\le {\tau(G)} + \log_2(\tau(G)) +1 $$
\end{lemma}

\begin{proof}
Let $G=(V,E)$ be a graph of order $n$, and let $S$ be an independent set of size $\alpha = n-\tau(G)$, and $R\subseteq S$ a subset of size $k$ to be fixed later. 
Let $G'=(R, (V{\setminus} S) \cup R,E')$ be a bipartite graph  s.t. for any $u\in R$, $N_{G'}(u) = \{u\}\cup N_G(u)$. Notice that there are two copies of $R$ in $G'$, one on each side of the bipartite graph: there is a matching between these two copies of $R$, the other edges of $G'$ are those of $G$ between $R$ and $V\setminus S$. According to lemma \ref{lem:odd} there exists $D\subseteq R'$ s.t $$|Odd_{G'}(D)|\le  \frac{|V| - |S| + |R|}{2(1-2^{-|R|})} = \frac{\tau(G)+k}{2(1-2^{-k})}$$ The odd-neighbourhood of $D$ in $G'$ is related to the odd-neighbourhood of $D$ in $G$ as follows: $Odd_{G'}(D) = \Delta_{u\in D}N_{G'}(u) = \Delta_{u\in D}(\{u\}\cup N_{G}(u)) = D\Delta Odd_{G}(D)$. Thus $|Odd_{G'}(D)| = |D\cup Odd_G(D)|$. As a consequence, $\delta_{loc}(G)+1\le \frac{\tau(G)+k}{2(1-2^{-k})}$. 
\begin{itemize}
\item If $\lceil \log_2(\tau(G)+1)\rceil \le n-\tau(G)$, then we fix $k=\lceil \log_2(\tau(G)+1)\rceil$: 
\begin{equation}
\delta_{loc}(G)+1\le \frac{\tau(G)+\lceil \log_2(\tau(G)+1)\rceil }{2(1-2^{-\lceil \log_2(\tau(G)+1)\rceil })}< \frac12(\tau(G)+\log_2(\tau(G))) + 1\label{eqn1}
\end{equation}
To prove the second inequality of equation (\ref{eqn1}), let $\tau(G) = 2^r+y$ with $y<2^r$. Notice that $\lceil \log_2(\tau(G)+1)\rceil = r+1$, thus 
\begin{eqnarray*}
\delta_{loc}(G)+1&\le& \frac{2^r+y+r+1}{2(1-2^{-r-1})}\\
\end{eqnarray*}
Moreover, standard calculation shows that $\frac{2^r+y+r+1}{1-2^{-r-1}}< 2^r+y+\log_2(2^r+y) +2$ when $r>0$. Thus $2\delta_{loc}(G)+2 < \tau(G)+\log_2(\tau(G)) +2$. When $r=0$, $\tau(G)=1$, thus $G$ is a star (and possibly some isolated vertices), so $2\delta_{loc}(G)\le 2 = \tau(G)+\log_2(\tau(G))+1$.

\item If $\lceil \log_2(\tau(G)+1)\rceil > n-\tau(G)$, then it is enough to prove that $2\delta_{loc}(G) \le n$ since $\tau(G) + \log_2(\tau(G)) +1 \ge \tau(G)+\lceil \log_2(\tau(G)+1)\rceil  >n $. For any set $S$ of size $\lfloor \frac n2\rfloor + 1$, $\cutrk_G(S)<|S|$ since $|V\setminus S|<|S|$, thus  according to property \ref{prop:hmp}, $\delta_{loc}(G)<\lfloor \frac n2\rfloor + 1\le n/2$. \hfill $\Box$

\end{itemize}

\end{proof}

\begin{remark} In Lemma \ref{lem:tau}, the condition $\tau(G)>0$ only excludes the empty graph and is used to guarantee that the logarithm is well defined. The bound is tight for star graphs: $\delta_{loc}(S_n)=1$ and $\tau(S_n)=1$. This is the only tight case and when $\tau(G)>1$, the proof can  be modified to prove the following statement where the constant factor is removed: if $\tau(G)>1$, $2\delta_{loc}(G)\le \tau(G)+\log_2(\tau(G))$. 
\end{remark}

The  vertex cover number-based bound on the local minimal degree leads to an improved general upper bound for  bipartite graphs:

\begin{theorem}\label{thm:delta_loc_bip}
For any bipartite graph $G$ of order $n>0$, $$\delta_{loc}(G)< \frac n4 + \log_2 n $$
\end{theorem}

\begin{proof} If $n\le 2$, the property is satisfied. Otherwise, since $G$ is bipartite $\tau(G)\le \lfloor \frac n2\rfloor$, so according to Lemma \ref{lem:tau}, 
$\delta_{loc}(G) \le   \frac 12 ( \tau(G)+\log_2(\tau(G))+1) \le  \frac n4 +\frac 12\log_2(n/2) +\frac 12 \le \frac n4 + \frac 12\log_2 n < \frac n4 + \log_2 n$.
\hfill $\Box$
\end{proof}

Contrary to the bipartite case, the bound involving the vertex cover number does not lead to an improved upper bound for non-bipartite graphs. However, we prove that the local minimum degree of a graph of order $n$ is at most $\frac 38n+o(n)$ exploiting the structure of the kernels of the linear maps associated with the cuts of the graph:

\begin{theorem}\label{thm:delta_loc}
For any graph $G$ of order $n>0$, $$\delta_{loc}(G) < \frac38n  + \log_2 n$$
\end{theorem}

\begin{proof} For any integer $0<k<n/2$, let $S$ be a subset of $\lfloor n/2 \rfloor + k$ vertices.  Let $L : S\to V\setminus S$ be the map $D\mapsto Odd_G(D) \setminus S$ which is linear for the symmetric difference, i.e. $L(D_1\Delta D_2) = L(D_1)\Delta L(D_2)$. Notice that for any $D\in Ker (L)$, $D\cup Odd(D) \subseteq S$. According to the rank nullity theorem, $dim(Ker(L)) \ge 2k -1$. Let $R\subseteq S$ be a basis for $Ker(L)$. Let $G'=(R, S\times \{1,2,3\},E')$ be a bipartite graph s.t. for any $D\in R, N_{G'}(D) = D{\times} \{1\} \cup Odd_G(D){\times} \{2\}\cup (Odd_G(D) \Delta D){\times} \{3\}$: the neighbourhood  of $D$ in $G'$ is the disjoint union of $D$, $Odd_G(D)$ and $D\Delta Odd_G(D)$. Notice that $|R| \ge 2k-1$ and $| S\times \{1,2,3\}| = 3(\lfloor n/2\rfloor+k)$, so according to lemma \ref{lem:odd}, 
 there exists a non empty $R_0\subseteq R$ such that  
$
|Odd_{G'}(R_0)| \le \left\lfloor \frac32 .\frac{\lfloor n/2 \rfloor +k}{1-2^{-2k+1}}\right\rfloor$.
 
 Let $F:= \Delta_{D\in R_0} D$. Since $R$ is a basis and $R_0\neq \emptyset$, $F\neq \emptyset$. Moreover $Odd_{G'}(R_0) = \Delta_{D\in R_0} N_{G'}(D) = \Delta_{D\in R_0} (D\times \{1\} \cup Odd_G(D)\times \{2\} \cup (Odd_G(D)\Delta D)\times \{3\}) = F\times \{1\} \cup Odd_{G}(F)\times \{2\} \cup (F\Delta Odd_G(F))\times \{3\}$. Thus  $|Odd_{G'}(R_0)| = |F| + |Odd_G(F)| + |F\Delta Odd_(F)| = 2|F\cup Odd_G(F)|$. As a consequence, 
 \begin{eqnarray}\label{eq:eq1}|F\cup Odd_G(F)|&\le&  \left\lfloor \frac12 \left\lfloor \frac32 .\frac{\lfloor n/2 \rfloor +k}{1-2^{1-2k}}\right\rfloor \right\rfloor
 \end{eqnarray}
 
 We choose $k{=}\lfloor4\log_2(n)/3\rfloor$ to guarantee $|F\cup Odd_G(F)| \le  \frac38n{+}\log_2(n){+}O(1)$. More precisely,  notice that $ |F\cup Odd_G(F)| {\le} \frac38. \frac{n  +  2\lfloor4\log_2(n)/3\rfloor}{1-2{\times} 2^{-2\lfloor4\log_2(n)/3\rfloor}} {\le}  \frac38. \frac{n  +  8\log_2(n)/3}{1-8.n^{-8/3}}$ 
  which is strictly smaller than  $\frac38n+\log_2 n + 1$ when $n>60$. For $2<n\le61$, one can double check by direct calculation that the bound in equation \ref{eq:eq1} is actually strictly smaller than $\frac38n+\log_2(n)+1$.  Thus for any $n>2$, $\min_{D\neq \emptyset}|D\cup Odd_G(D)| < \frac38n+\log_2 n + 1$, so $\delta_{loc}(G)< \frac38 n+\log_2 n$. Finally, it is easy to check that  $\delta_{loc}(G)< \frac38 n+\log_2 n $ also holds for $n\le 2$. \hfill $\Box$
 \end{proof}

\begin{remark}
Choosing $k = \lfloor \log_2(n)/2\rfloor$ in the proof of theorem \ref{thm:delta_loc} gives an asymptotically slightly better bound: $\delta_{loc}(G)\le 3/8 n+3/4\log_2 (n) + O(1)$. 
\end{remark}

\section{Parameterized Complexity}\label{sec:FPT}

The decision problem associated with the local minimum degree is known to be NP-complete and hard to approximate: there exists no $k$-approximation algorithm for this problem for any constant $k$ unless P=NP \cite{JMP-WG12}. In this section we consider the parameterized complexity of this problem, and its bipartite version. Please refer to \cite{monograph} for an introduction to parameterized complexity.
~\\

\noindent \begin{tabular}{ll}
\textsc{Local Minimum Degree}:~~~~~~~~~~~~~& \textsc{Bipartite Local Minimum Degree}:\\
\emph{input:} A  graph $G$&\emph{input:} A  bipartite graph $G$\\
\emph{parameter:} An integer $k$&\emph{parameter:} An integer $k$\\
\emph{question:} Is $\delta_{loc}(G) \leq k$?&\emph{question:} Is $\delta_{loc}(G) \leq k$?
\end{tabular}

\vspace{0.2cm}

\noindent We show  that both problems are  FPT-equivalent to the \textsc{EvenSet} problem \cite{oddset}:
\vspace{-0.2cm}

\noindent \textsc{EvenSet}:\\ 
\emph{input:} A  bipartite graph $G=(R,B,E)$\\
\emph{parameter:} An integer $k$\\
\emph{question:} Is there a non empty $D\subseteq R$, such that $|D|\le k$ and $Odd_G(D)=\emptyset$ i.e., every vertex in $B$ has an even number of neighbours in $D$? \vspace{0.2cm}

\noindent To prove the FPT-equivalence of these three problems, first  we prove  that \textsc{EvenSet} is harder than \textsc{Local Minimum Degree}, and then that \textsc{Bipartite Local Minimum Degree} is harder than \textsc{EvenSet}.

\begin{theorem}
\label{lem:evenSet<DeltaLoc}
\textsc{EvenSet} is  \textup{FPT}-reducible to \textsc{Local Minimum Degree}.
\end{theorem}

\begin{proof}
Given an instance $(G,k)$ of \textsc{Local Minimum Degree}, let $(G',k')$ be an instance of
\textsc{EvenSet} where:\\
$G'=(A_1\cup A_2, \cup A_3, A_4\cup A_5,E_1 \cup E_2 \cup E_3)$, $k'=2k{+}2$\\
$\forall i\in [1,5], A_i = \{a_{i,u}, \forall u \in V(G)\}$\\
$E_1=\{(a_{1,u},a_{4,u}), \forall u \in V(G) \}$,\\
$E_2=\{(a_{i,u},a_{5,u}), \forall i \in \{2,3\}, \forall u \in V(G) \}$\\
$E_3=\{(a_{2,u},a_{i,v}), \forall i \in \{4,5\}, \forall \{u,v\} \in E(G) \}$\\
In other words,  $G'$ consists of  5 copies $A_i$s of $V(G)$, there is a matching between $A_1$ and $A_4$, and between $A_3$ and $A_5$. Moreover, the  subgraph induced by $A_2\cup A_4$ is the bipartite double of $G$, whereas  subgraph induced by $A_2\cup A_5$ the bipartite double of $G$ augmented with a matching.

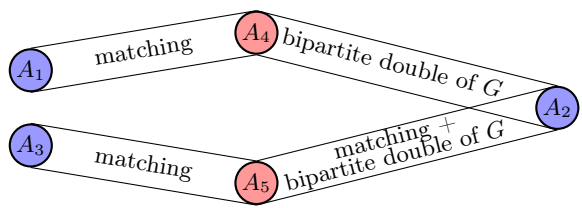
\centerline{\footnotesize\begin{tikzpicture}[xscale=0.50,yscale=0.50,every node/.style={draw=black,thick,circle,inner sep=1pt}]
\useasboundingbox (-2,-3) rectangle (2,3);
\node [draw=none] at (-3,1.5) {\rotatebox{9.5}{matching}};
\node [draw=none] at (-3,-1.5) {\rotatebox{-9.5}{matching}};
\node [draw=none] at (3.6,1.15) {\rotatebox{-14}{bipartite double of $G$}};
\node [draw=none] at (3.7,-1.32) {\rotatebox{14}{bipartite double of $G$}};
\node [draw=none] at (3.7,-0.8) {\rotatebox{14}{matching +}};
\node[draw,fill=blue!40] (A1) at (-6,1) {$A_1$};
\node[draw,fill=blue!40] (A2) at (8,0){$A_2$};
\node[draw,fill=blue!40,circle] (A3) at (-6,-1){$A_3$};
\node[draw,fill=red!40,circle] (A4) at (0,2){$A_4$};
\node[draw,fill=red!40,circle] (A5) at (0,-2){$A_5$};
\draw (A1.north) -- (A4.north);
\draw (A1.south) -- (A4.south);
\draw (A2.north) -- (A4.north);
\draw (A2.south) -- (A4.south);
\draw (A2.north) -- (A5.north);
\draw (A2.south) -- (A5.south);
\draw (A3.north) -- (A5.north);
\draw (A3.south) -- (A5.south);
\end{tikzpicture}}

\noindent -- If $(G,k)$ is a positive instance of \textsc{Local Minimum Degree} with a non empty $D\subseteq V(G)$ such that $|D\cup Odd_G(D)|\le k{+}1$. 
Let $D' = \{a_{1,u}~|~u\in Odd_G(D)\}\cup \{a_{2,u}~|~u\in D\}\cup\{a_{3,u}~|~u\in Odd_G(D)\Delta D\}$, thus $D'$ is composed of the copy of $D$ in $A_2$, the copy of $Odd_G(D)$ in $A_1$ and the copy of $D\Delta Odd_G(D)$ in $A_3$. Notice that $Odd_{G'}(D') = \emptyset$, and $D'\neq \emptyset $ since $D\neq \emptyset$. Moreover $|D'| = |Odd_G(D)|+|D|+|D\Delta Odd_G(D)| = 2|D\cup Odd_G(D)|\le 2k +2= k'$. Thus $D'$ makes $(G',k')$ a positive instance of \textsc{EvenSet}. \\
\noindent -- If $(G',k')$ is a positive instance of \textsc{EvenSet} with a non empty $D\subseteq A_1\cup A_2\cup A_3$ of size at most  $k'$ such that $Odd_{G'}(D) = \emptyset$. For $i\in [1,3]$, let $D_i = \{u\in V(G) ~|~a_{i,u} \in D\}$. Notice that $D_1 = Odd_G(D_2)$ and $D_3 = Odd_G(D_2)\Delta D_2$.  $D\neq \emptyset$ implies  $D_2\neq \emptyset$, moreover $|D_2\cup Odd_G(D_2)| =\frac12( |D_2|+ |Odd_G(D_2)| + |Odd_G(D_2)\Delta D_2|) = \frac12 |D| \le \frac 12k' = k{+}1$, so $D_2$ makes $(G,k)$ a positive instance of  \textsc{Local Minimum Degree}. \hfill $\Box$
\end{proof}

\begin{corollary}
\textsc{Local Minimum Degree} is in \textup{W[2]}.
\end{corollary}

W[2]-membership of \textsc{Local Minimum Degree} is not surprising in the sense that  not only \textsc{EvenSet} but all similar problems of graph domination with parity conditions are known to be  in W[2] \cite{CP-TAMC14}. 
We refine this W[2]-membership by proving  that both \textsc{Local Minimum Degree} and \textsc{Bipartite Local Minimum Degree} are FPT-equivalent to \textsc{EvenSet}. They form a peculiar subclass of W[2] for which no hardness results are known: the W[1]-hardness of \textsc{EvenSet} is  a long standing open question in parameterized complexity \cite{oddset}. This contrasts with the subclass of problems FPT-equivalent to the W[1]-hard  \textsc{OddSet} problem which contains problems like \textsc{Weak Odd Domination}   and \textsc{Quantum Threshold} \cite{CP13-FCT,GJMP15}. 

\begin{theorem}
\label{lem:deltaLoc<EvenSet}
\textsc{Bipartite Local Minimum Degree} is \textup{FPT}-reducible to \textsc{EvenSet}.
\end{theorem}
\begin{proof}
If $(G{=}(R,B,E),k)$ is a positive instance of \textsc{EvenSet}, then it is also a positive instance of \textsc{Bipartite Local Minimum Degree}. But  if  $(G,k)$ is a positive instance of \textsc{Bipartite Local Minimum Degree}, it may fail to be a positive instance of \textsc{EvenSet} mainly for two reasons: \\
 (i) A set $D$ such that $|D\cup Odd_G(D)|\le k{+}1$ may not be a subset of $R$\\
 (ii) For solving  \textsc{EvenSet}, one wants to guarantee that $Odd_G(D)=\emptyset$.

Regarding the first point, a gadget  with a local minimum degree larger than $k{+}1$ is attached to each vertex in $B$ to guarantee that no vertex of $B$ can occur in a set $D$ such that $|D\cup Odd(D)|\le k{+}1$. Concretely we can use a Paley graph $P_q$ which vertices are $\{0, \ldots, q-1\}$ for $q=1\bmod 4$ a power of prime, and $(i,j)$ is an edge iff $\exists x, i-j= x^2\bmod q$.  The local minimal degree of a Paley graph is at least  square root of its order. However to keep the bipartiteness of the graph we use the bipartite double of a Paley graph rather than a Paley graph. Indeed, it is known that the local minimum degree of a bipartite double  graph  is as large as the local minimum degree of the original graph ($\delta_{loc} (G^{\oplus 2}) \ge \delta_{loc}(G)$ \cite{Javelle-these}).

Regarding the second point, each vertex of $B$ is duplicated $k$ times in such a way that for any $D\subseteq R$ if a vertex $v\in B$ is in the odd neighbourhood of $D$ than its $k$ copies are also in the odd-neighbourhood which contradicts the fact that $|D\cup Odd(D)|$ is at most  $k+1$. 

Concretely, let $q$ be a prime number such that $k^2<q\le 2k^2+5$ and $q=1\bmod 4$. Such a prime number $q$ exists \cite{Cullinan} and can be found in time polynomial in $k$. Let $(G',k)$ be an instance of \textsc{Bipartite Local Minimum Degree} such that\\
$G'=(R\cup P', P$, $E_{G}\cup E_{\text{Paley}})$,  
where 
$P=\cup_{b\in B, i\in [0,k]} P_{b,i}$, $P'=\cup_{b\in B, i\in [0,k]} P'_{b,i}$  
$P_{b,i} {=} \{p_{b,i,r}, \forall r{\in} [0, q-1]\}$, $P'_{b,i} {=} \{p'_{b,i,r}, \forall r{\in} [0, q-1]\}$  
$E_{\text{Paley}}{=}\cup_{b\in B, i\in[0,k]} E^{(b,i)}_{\text{Paley}}$ and 
$E^{(b,i)}_{\text{Paley}} {=} \{(p_{b,i,r}, p'_{b,i,r'}),  \forall r,r'{\in} [0{,}q-1] ~s.t.~ \exists  \ell {\in} [0,q-1], \ell^2 {=} r{-}r'\bmod q\}$. \\
-- If $(G,k)$ is a positive instance of \textsc{EvenSet} with $D{\subseteq} E$ s.t. $Odd_G(D) {=} \emptyset$ then $Odd_{G'}(D) {=} \emptyset$ so $(G',k)$ is a positive instance of \textsc{Bipartite Local Minimum Degree}. \\
-- If $(G',k)$ is a positive instance of \textsc{Bipartite Local Minimum Degree} with $D$ s.t. $|D\cup Odd_{G'}(D)|\le k{+}1$. For any $b\in B, i\in [0,k]$, let $D'_{b,i} = D\cap (P_{b,i}\cup P'_{b,i})$, 
in the subgraph induced by $P_{b,i}\cup P'_{b,i}$ $|D' \cup Odd_{G'[P_{b,i}\cup P'_{b,i}]}(D)|\le k+1$, thus $D'_{n,i}= \emptyset$ since $\delta_{loc}(\text{Paley}_{k^2+1})> k$. 
So $D\subseteq R$. Moreover if there exists  
$p_{b,i,0}\in Odd_{G'}(D)$ then $\forall j\in [0,k], p_{b,j,0}\in Odd_{G'}(D)$, 
so $|D\cup Odd_{G'}(D)|>k{+}1$, so by contradiction $Odd_{G'}(D){=}\emptyset$. Thus $(G,k)$ is a positive of \textsc{EvenSet}. 
\hfill $\Box$
\end{proof}

\begin{corollary}
 \textsc{Bipartite Local Minimum Degree} and \textsc{Local Minimum Degree} are \textup{FPT}-equivalent to  \textsc{EvenSet}.
 \end{corollary}

W[1]-hardness of \textsc{EvenSet} is a long standing open problem, the FPT-equivalence with \textsc{(Bipartite) Local Minimum Degree} might give some more insights and open new perspectives on the parameterized complexity of \textsc{EvenSet}.

\section{Exponential algorithms}\label{sec:expo}

In this section we introduce exact exponential algorithms for computing the local minimum degree of a graph. 

\begin{property}
The local minimum degree  of a graph of order $n$ can be computed in time $\mathcal{O}^*(1.938^n)$.
\end{property}

\begin{proof}
Thanks to Property \ref{prop:hmp} and Theorem \ref{thm:delta_loc}, $\delta_{loc}(G) {+}1 {=} \min\{|A| :  |A| \le  \frac 3 8 n + \log_2(n) \wedge \cutrk_G(A){<}|A|\}$. The algorithm consists in enumerating all  subsets of at most $\frac 38n{+}\log_2(n)$ vertices and computing its cut-rank. The cut-rank can be computed in polynomial time, so the complexity of this algorithm is $\mathcal{O}^*(2^{H(\frac 38)n})$ where $H(x){=}-x\log_2 x -(1{-}x)\log_2(1{-}x)$ is the binary entropy function.\hfill $\Box$

\end{proof}

Regarding the bipartite case, enumerating all the subsets of size at most $\frac n4+\log_2(n)$ leads to a $\mathcal{O}^*(1.755^n)$ algorithm. This naive algorithm can be improved:

\begin{theorem}
The local minimum degree of a bipartite graph of order $n$ can be computed in time $\mathcal{O}^*(1.466^n)$.
\end{theorem}

\begin{proof}
We use the following property of bipartite graphs: given a bipartite graph $G = (V_1,V_2,E)$, $\delta_{loc}(G)  +1 =\min_{\emptyset\subset D\subseteq V_1\text{ or }\emptyset\subset D\subseteq V_2} |D\cup Odd_G(D)|$. 
Indeed, for any $D\subseteq V_1\cup V_2$, both $(D\cap V_1)\cup Odd_G(D\cap V_1)$ and $(D\cap V_2)\cup Odd_G(D\cap V_2)$ are subsets of $ D\cup Odd_G(D)$. Let  $|V_1| = \alpha n$ and $|V_2| = (1-\alpha)n$. We assume w.l.o.g. that $\alpha\le 1/2$. 
Since $V_1$ is a vertex cover set, according to Lemma \ref{lem:tau}, $\delta_{loc}(G)\le \frac \alpha 2 n +\frac{\log_2(\alpha n)}2$. 
Thus to compute the local miminum degree, it is enough to enumerate all sets $D$ of size at most $\frac \alpha 2 n +\frac{\log_2(\alpha n)}2$ in both $V_1$ and $V_2$ and to compute their odd neighbourhood -- which can be done in  time polynomial in $n$. 
There are ${\alpha n \choose  \frac \alpha 2 n +\frac{\log_2(\alpha n)}2} + {(1-\alpha) n \choose  \frac \alpha 2 n +\frac{\log_2(\alpha n)}2} = \mathcal O^*(2^{(1-\alpha)nH(\frac{\alpha}{2(1-\alpha)})})$ sets to enumerate. Notice that $\alpha \mapsto (1-\alpha)H(\frac{\alpha}{2(1-\alpha)})$ is maximal for $\alpha_0=0.3885$, and $2^{(1-\alpha_0)H(\frac{\alpha_0}{2(1-\alpha_0)})}=1.46557$. \hfill $\Box$
\end{proof}
\section{Conclusion}

After having shown that the local minimum degree is smaller than half of the vertex cover number (up to a logarithmic term), we have improved the best known upper bound on the local minimum degree, proving that it is at most $\frac 38n+o(n)$ and $\frac n4+o(n)$ for bipartite graphs. Moreover, we have investigated the parametrized complexity of the problem, showing  its W[2]-membership and its FPT-equivalence with the \textsc{EvenSet} problem, even when restricted to bipartite graphs. Finally, we have introduced a $\mathcal O^*(1.938^n)$-algorithm -- $\mathcal O^*(1.466^n)$-algorithm for the bipartite graphs -- for computing the local minimum degree. 

This is noticeable that the bipartite case evolves quite similarly to the general case: same parameterized complexity, and upper bound and algorithm slightly better in the bipartite case. It would be interesting to investigate other families of graphs, in particular those defined by excluded vertex minors, in order to identify a family of graphs which local minimum is large but easy to compute or to approximate. 

\subsubsection*{Acknowledgments. } 
We would like to thank Emmanuel Jeandel and Mehdi Mhalla for several helpful discussions. This work has been partially funded by the ANR-10-JCJC-0208 CausaQ grant and by r\'egion Rh\^one-Alpes  (ADR Cible R637).


\begin{thebibliography}{10}

\bibitem{ChuangShor}
S.~Beigi, I.~Chuang, M.~Grassl, P.~Shor, and B.~Zeng.
\newblock Graph concatenation for quantum codes.
\newblock {\em Journal of Mathematical Physics}, 52(2)(022201), 2011.

\bibitem{Bouchet87}
A.~Bouchet.
\newblock Graphic presentations of isotropic systems.
\newblock {\em J. Comb. Theory Ser. A}, 45:58--76, July 1987.

\bibitem{Bouchet85-connectivity}
A.~Bouchet.
\newblock Connectivity of isotropic systems.
\newblock In New York~Academy of~Sciences, editor, {\em Proceedings of the
  third international conference on Combinatorial mathematics}, pages 81--93,
  1989.

\bibitem{Bouchet90-kappa}
A.~Bouchet.
\newblock $\kappa$-transformations, local complementations and switching.
\newblock In {\em NATO Adv. Res. Workshop}, volume~C, pages 41--50, 1990.

\bibitem{Bouchet94-circle}
A.~Bouchet.
\newblock Circle graph obstructions.
\newblock {\em J. Comb. Theory, Ser. B}, 60(1):107--144, 1994.

\bibitem{BFK08}
A.~Broadbent, J.~Fitzsimons, and E.~Kashefi.
\newblock Universal blind quantum computation.
\newblock In {\em 50th Annual IEEE Symposium on Foundations of Computer
  Science, FOCS 2009}, 2009.

\bibitem{CP13-FCT}
D.~Cattan{\'e}o and S.~Perdrix.
\newblock Parametrized complexity of weak odd domination problems.
\newblock In {\em 19th International Symposium on Fundamentals of Computation
  Theory (FCT'13)},  {\em LNCS} vol 8070,
  pp 107--120. Springer, 2013.

\bibitem{CP-TAMC14}
D.~Cattan{\'e}o and S.~Perdrix.
\newblock The Parameterized Complexity of Domination-type Problems and Application to Linear Codes. 
\newblock In {\em  11th Annual Conference on Theory and Applications of Models of Computation
  Theory (TAMC'14)}, {\em LNCS} vol. 8402, 
  pages 86--103. Springer, 2014.

\bibitem{Cullinan}
J.~Cullinan and F.~Hajir. 
\newblock Primes of prescribed congruence class in short intervals. 
\newblock Integers, 12, A56. 2012. 
 

\bibitem{deFraysseix81}
H.~de~Fraysseix.
\newblock Local complementation and interlacement graphs.
\newblock {\em Discrete Mathematics}, 33(1):29--35, 1981.

\bibitem{monograph}
R.G. Downey and M.R. Fellows.
\newblock {\em Parameterized Complexity}.
\newblock Springer-Verlag, 1999.

\bibitem{oddset}
R.G. Downey, M.R. Fellows, A.~Vardy, and G.~Whittle.
\newblock The parameterized complexity of some fundabmental problems in coding
  theory.
\newblock {\em CDMTCS Research Report Series}, 1997.



\bibitem{fon1996local}
D.G.~Fon-Der-Flaasss.
\newblock Local complementations of simple and directed graphs.
\newblock In {\em Discrete analysis and operations research}, pages 15--34.
  Springer, 1996.

\bibitem{GJMP-MEMICS12}
S.~Gravier, J.~ Javelle, M.~Mhalla, and S.~Perdrix.
\newblock Quantum secret sharing with graph states.
\newblock In proceedings of {\em MEMICS}, vol. 7721 of {\em
  LNCS}, pp 15--31. Springer, 2012.


\bibitem{GJMP15}
S.~Gravier, J.~ Javelle, M.~Mhalla, and S.~Perdrix.
\newblock On Weak Odd Domination and Graph-based Quantum Secret Sharing.
\newblock {\em Theor. Comput. Sci.} 598: 129-137, 2015.


\bibitem{HEB04}
M.~Hein, J.~Eisert, and H.~J. Briegel.
\newblock Multi-party entanglement in graph states.
\newblock {\em Physical Review A}, 69:062311, 2004.

\bibitem{HMP06}
P.~H\/{\o}yer, M.~Mhalla, and S.~Perdrix.
\newblock Resources required for preparing graph states.
\newblock In {\em Proceedings of ISAAC06, LNCS}, volume 4288, pages 638--649,
  2006.

\bibitem{Javelle-these}
J.~Javelle.
\newblock {\em Cryptographie Quantique, Protocoles et Graphes}.
\newblock PhD thesis, Grenoble University, 2014.

\bibitem{JMP-WG12}
J.~Javelle, M.~Mhalla, and S.~Perdrix.
\newblock On the minimum degree up to local complementation: Bounds and
  complexity.
\newblock In proceedings of {\em Graph-Theoretic Concepts in Computer Science
  (WG)},  of {\em LNCS} vol. 7551, pp
  138--147. Springer, 2012.

\bibitem{CJWY07}
Z.~Ji, J.~Chen, Z.~Wei, and M.~Ying.
\newblock The lu-lc conjecture is false, 2007.

\bibitem{kotzig}
A.~Kotzig.
\newblock Eulerian lines in finite 4-valent graphs and their transformations.
\newblock In {\em Colloqium on Graph Theory Tihany 1966}, pages 219--230.
  Academic Press, 1968.

\bibitem{MS08}
D.~Markham and B.C.~Sanders.
\newblock Graph states for quantum secret sharing.
\newblock {\em Physical Review A}, 78:042309, 2008.

\bibitem{Oum08}
S.~Oum.
\newblock Approximating rank-width and clique-width quickly.
\newblock {\em ACM Transactions on Algorithms}, 5(1), 2008.

\bibitem{plotkin60}
M.~ Plotkin.
\newblock Binary codes with specified minimum distance.
\newblock {\em Information Theory, IRE Transactions on}, 6(4):445--450, 1960.

\bibitem{RB01}
R.~Raussendorf and H.J.~Briegel.
\newblock A one-way quantum computer.
\newblock {\em Phys. Rev. Lett.}, 86:5188--5191, 2001.

\bibitem{Schling05}
D.~Schlingemann.
\newblock Local equivalence of graph states.
\newblock In O.Krueger and R.F.Werner, editor, {\em Some Open Problems in Quantum
  Information Theory},  arXiv:quant-ph/0504166, 2005.

\bibitem{VdN05}
M.~Van~den Nest.
\newblock {\em Local equivalence of stabilizer states and codes}.
\newblock PhD thesis, Faculty of Engineering, K.\,U. Leuven, Belgium, May 2005.

\bibitem{VdN04}
M.~Van~den Nest, J.~Dehaene, and B.~De Moor.
\newblock Graphical description of the action of local clifford transformations
  on graph states.
\newblock {\em Phys. Rev. A}, 69:022316, 2004.

\bibitem{zeng07}
B.~Zeng, H.~Chung, A.W~Cross, and I.L.~Chuang.
\newblock Local unitary versus local clifford equivalence of stabilizer and
  graph states.
\newblock {\em Phys. Rev. A}, 75(3):032325, 2007.

\end{thebibliography}
\end{document}